\documentclass[11pt,a4paper,showkeys]{article}

\usepackage{amsmath}
\usepackage{amssymb}
\usepackage{amsthm}
\usepackage{latexsym}
\usepackage[dvipdfmx]{graphicx}
\usepackage[dvips]{color}

\theoremstyle{plain}
\newtheorem{theorem}{Theorem}[section]
\newtheorem{corollary}[theorem]{Corollary}

\newtheorem{lemma}[theorem]{Lemma}

\theoremstyle{definition}

\newcommand{\ZZ}{{\mathbb{Z}}}
\newcommand{\RR}{{\mathbb{R}}}
\newcommand{\Ueven}{U_{\mathrm{even}}}
\newcommand{\Uodd}{U_{\mathrm{odd}}}
\newcommand{\Veven}{V_{\mathrm{even}}}
\newcommand{\Vodd}{V_{\mathrm{odd}}}

\usepackage{fullpage}

\usepackage{newtxtext,newtxmath}

\title{On the Equivalence of the Graph-Structural and Optimization-Based Characterizations of Popular Matchings}

\author{
Yuga Kanaya\thanks{Department of Systems Engineering, Graduate School of Science and Engineering, Hosei University, Tokyo 184-8584, Japan.}
\and 
Kenjiro Takazawa\thanks{Department of Industrial and Systems Engineering, Faculty of Science and Engineering, 
Hosei University, Tokyo 184-8584, Japan.  
\texttt{takazawa@hosei.ac.jp}.  
Partially supported by 
JSPS KAKENHI Grant Numbers 
JP24K02901, JP24K14828, 
Japan.
}}

\date{August 2026}

\begin{document}

\maketitle

\begin{abstract}
In a bipartite graph in which the vertices have preferences over their neighbours, a matching is defined to be popular if it does not lose in a majority vote against any matching. In this paper, we study the following three primary problems: only the vertices on one side have preferences; a generalization of this problem allowing ties in the preferences; and the vertices on both sides have preferences. 
A principal issue in the algorithmic aspects of popular matchings is how to determine the popularity of a matching. 
In the literature, we have the following two types of characterizations. 
The graph-structural characterizations are specifically designed for each problem and provide a combinatorial structure of the popular matchings. The optimization-based characterizations are described by maximum-weight matchings and work in the same manner for all problems, while they do not reveal the structure of the popular matchings. 
A main contribution of this paper is to provide a direct connection of these two characterizations for all of the three problems. Specifically, we prove that each characterization can be derived from the other, without relying on either algorithms or the fact that they characterize popular matchings. Our proofs offer a comprehensive understanding of the equivalence of the two characterizations, and suggest a new interpretation of the graph-structural characterization in terms of the dual optimal solution for the maximum-weight matching problem. Further, this interpretation suggests a possibility of deriving a graph-structural characterization for problems of popular matchings for which only the optimization-based characterization is known. 
\end{abstract}

\section{Introduction}

\label{SECintro}

\emph{Popular matchings} \cite{Gar75} provide a model of matching under preferences \cite{Man13} which has relevance to computational social choice, 
and  has similarity to both stable matchings and votings. 
One typical problem of popular matchings is described as follows. 

Let $G=(U,V;E)$ denote a bipartite graph with color classes $U,V$ and edge set $E$. 
An edge set $M\subseteq E$ is called a \emph{matching} if 
each vertex $u\in U\cup V$ is incident to at most one edge in $M$. 
For a vertex $z$ and a matching $M$, 
let $M(z)$ denote the vertex which is matched with $z$ by $M$. 
Note that $M(z)$ may be empty. 

For each vertex $z\in U\cup V$, 
let $\Gamma(z)\subseteq U\cup V$ denote the set of vertices adjacent to $z$, 
i.e.,\ 
$\Gamma(z) =\{ z' \in U\cup V \colon (z,z')\in E\}$. 
Each vertex $z$ has \emph{preferences} over the vertices in $\Gamma(z)$, 
which is represented by a totally ordered set $(\Gamma(z), \succ_z)$. 
Namely, for $z',z''\in \Gamma(z)$, 
$z' \succ_z z''$ denotes that $z$ prefers $z'$ to $z''$. 
We assume that every vertex $z$ prefers every vertex in $\Gamma(z)$ to the empty set. 

For two matchings $M,N\subseteq E$, 
Let $\Delta(M,N)$ denote the number of vertices which prefers $M$ to $N$. 
Formally, 
\begin{align*}
\Delta (M,N) = |\{z\in U\cup V \colon M(z) \succ_u N(z)\}| - |\{z\in U\cup V \colon N(z) \succ_z M(z)\}|. 
\end{align*}
A matching $M$ is \emph{popular} if 
$\Delta(M,N) \ge 0$ for each matching $N$ in $G$. 

The above definition immediately leads to one reasonable interpretation of popular matchings. 
Namely, a popular matching is a Condorcet winner of the voting system 
in which the vertices in $U\cup V$ elect a matching in $G$. 
Popular matchings also have a close connection to stable matchings. 
%%Chung \cite{Chu00} proved that a stable matching is a popular matching. 
G\"ardenfors \cite{Gar75} proved that a stable matching is a popular matching. 
More specifically, 
Huang and Kavitha \cite{HK13} proved that a stable matching is a popular matching of minimum size. 
Namely, 
another reasonable interpretation of popular matchings is that 
they are matchings reflecting the preferences of the agents, 
while having the possibility of including more agents than stable matchings.

Similarly to stable matchings, 
popular matching problems are classified into various problems. 
The problem described above is referred to as the \emph{stable matching problem with incomplete lists}, 
and \emph{SMI problem} for short. 
Other primary problems are the \emph{house allocation problem} (\emph{HA problem}), 
in which only the vertices in $U$ have preferences, 
and 
the \emph{house allocation problem with ties} (\emph{HAT problem}), 
which is a generalization of the HA problem allowing ties in the preferences. 
Precise definitions of those problems will be provided in Section \ref{SECpre}. 

A primary issue in the algorithmic aspects of popular matchings is 
how to determine whether a matching $M$ is popular or not. 
This is because, 
unlike stable matchings, 
the above definition of popular matchings does not directly imply an efficient method for determining popularity. 
Namely, 
if the above definition is simply used, 
the computation of $\Delta(M,N)$ for every other matching $N$ is required,  
which may take exponential time. 
Thus, 
providing a characterization of popular matchings which can be efficiently checked is essential in the algorithmic study of popular matchings. 

In this paper, 
we focus on two major types of those characterizations in the literature. 
A seminal work by Abraham, Irving, Kavitha, and Mehlhorn \cite{AIKM07} was the first to provide such characterizations 
for the HA and HAT problems, 
which lead to a polynomial-time algorithm for finding a popular matching. 
These characterizations are based on the structure of the graph and the preferences, 
and imply the combinatorial structure of all popular matchings. 
We refer to this type of characterizations as \emph{graph-structural characterizations}. 
A graph-structural characterization for the SMI problem is given by 
Huang and Kavitha \cite{HK13}, 
which also leads to a polynomial-time algorithm for finding a popular matching. 
We remark that 
the graph-structural characterizations are specifically designed for each problem.  

The other characterization is described by maximum-weight matchings with respect to edge weights defined from the preferences, 
which leads to a polynomial-time algorithm for determining whether a matching is popular. 
In contrast to the graph-structural characterization, 
this characterization works for all of the above problems in the same manner \cite{BIM10,KMN11,McC08}, 
while it does not reveal the structure of the popular matchings. 
We refer to this characterization as an \emph{optimization-based characterization}.

\subsection{Our contribution}

The aim of this paper is to 
investigate a direct connection between the graph-structural and optimization-based characterizations 
for the HA, HAT, and SMI problems. 
By definition, the two characterizations are equivalent in each problem. 
However, 
their direct connection has been unclear: 
the equivalence is explained in a way that they are characterizing the same fact 
(i.e.,\ a matching is popular), 
or 
is 
implicitly derived with the aid of algorithms in \cite{KKMSS21} (an arXiv version of \cite{KKMSS22soda}). 

Our contribution is to prove that the two characterizations are equivalent 
without relying on either algorithms or the fact that they are characterizing popular matchings. 
Namely, 
we show that a matching satisfying the graph-structural characterization 
also satisfies the optimization-based characterization, 
and vice versa, 
for the HA, HAT, and SMI problems. 

In the proof, 
we employ the theory of combinatorial optimization \cite{KV18,Sch03},  
in particular bipartite matchings \cite{LP86} and linear-programming duality \cite{Sch86}. 
Our technical ingredients include 
the duality theory for linear optimization (the strong duality theorem and complementarity slackness), 
the min-max formula for bipartite matchings (K\H{o}nig's theorem \cite{Kon31}), 
the total dual integrality of the linear system 
describing the bipartite matching polytope (Egerv\'ary's theorem \cite{Ege31}), 
the Dulmage-Mendelsohn decomposition \cite{DM58,DM59}, 
and 
alternating paths. 
Specifically, 
each characterization is derived from the other in the following way. 

\paragraph{Deriving the optimization-based characterization}
For all of the HA, HAT and SMI problems, 
the duality theory for linear optimization plays a key role. 
Namely, 
on the basis of the graph-structural characterization, 
we construct a dual feasible solution $y^*$ for the maximum-weight matching problem 
whose objective value is equal to the weight of the current matching $M$, 
proving the optimality of $M$. 
We remark that the obtained dual optimal solution $y^*$ has not only the integrality but also a certain property, 
which is of an independent interest and will be described in Corollaries \ref{CORHA}, \ref{CORHAT}, and \ref{CORSMI}. 
% For the HA and HAT problems, 
% $y^*$ is a $\{0,1\}$-vector, while the edge weight is a $\{0,1,2\}$-vector. 
% For the SMI problem, 
% $y^*$ is a $\{0,1,2\}$-vector, while the edge weight is a $\{0,1,2,3,4\}$-vector. 

\paragraph{Deriving the graph-structural characterization}
For the HA and HAT problems, 
there exists a dual integral optimal solution $y^*$ corresponding to the maximum-weight matching $M$, 
on the basis of 
%%the complementarity slackness and 
the dual integrality of the linear system \cite{Ege31}. 
The optimal solution $y^*$ is proven to be a $\{0,1\}$-vector, 
and leads to the graph-structural characterization. 
We remark that, 
in the HAT problem, 
the min-max theorem for bipartite matching and minimum cover \cite{Kon31} and 
the Dulmage-Mendelsohn decomposition \cite{DM58,DM59} play essential roles. 
For the SMI problem, 
the graph-structural characterization is derived from a combinatorial argument utilizing alternating paths. 

\bigskip

Overall, 
our proofs 
offer a comprehensive understanding that the two types of characterizations are 
essentially equivalent, 
and 
suggest a new interpretation of the graph-structural characterization 
in terms of the dual optimal solution for the maximum-weight matching problem. 
Since the optimization-based characterization is applicable to various problems of popular matchings, 
this interpretation suggests a possibility of deriving a graph-structural characterization 
for some problems of popular matchings for which only the optimization-based characterization is known (see Section \ref{SECcncl}). 

\subsection{Further related work}

Popular matchings can be studied in several principal settings other than those mentioned above. 
Typical problems are the \emph{stable roommate problem with incomplete lists} (\emph{SRI problem}) 
and the \emph{stable roommate problem with ties and incomplete lists} (\emph{SRTI problem}),  
in which the graph $G$ is nonbipartite. 
Chung \cite{Chu00} proved that a stable matching is popular in the SRI problem. 
Moreover, 
the graph-structural characterization for the SMI problem by Huang and Kavitha \cite{HK13} applies to the SRI problem. 
The optimization-based characterization due to Bir\'o et al.\ \cite{BIM10} was indeed given for the SRI and SRTI problems, 
which is followed by a more efficient characterization for the SRI problem by B\'{e}rczi-Kov\'{a}cs and Kosztol\'{a}nyi \cite{BKK25}. 
As far as we are aware, no prior work provides a graph-structural characterization for the SRTI problem. 

Other typical problems include the \emph{capacitated house allocation problem} \cite{MS06}, 
\emph{weighted house allocation problem} \cite{Mes14}, 
\emph{weighted capacitated house allocation problem} \cite{SM10}, 
\emph{popular $b$-matching problem} \cite{Csa24,Pal14}, 
and \emph{popular assignment problem} \cite{KKMSS22soda}. 
Further, 
problems which in particular relate to matroid intersection are 
the \emph{popular branching problem} \cite{KKMSS22,NT23}, 
\emph{popular arborescence problem} \cite{KMSY25}, 
and \emph{popular matching problem with matroid constraints} \cite{CKTY25+,Kam16,Kam17,Kam20}. 
A detailed discussion on these problems is given in Section \ref{SECcncl}.

\subsection{Organization of the paper}

%The rest of the paper is organized as follows. 
In Section \ref{SECpre}, 
we present prior results which are closely related to our work. 
In Sections \ref{SECHA}, \ref{SECHAT}, 
and \ref{SECSMI}, 
we derive the optimization-based characterization from the graph-structural characterization 
and also provide the reverse derivation 
for the HA, HAT, and SMI problems, 
respectively. 
Section \ref{SECcncl} concludes the paper with a summary of our work and possible directions of future work. 

\section{Preliminaries}
\label{SECpre}

\subsection{Basics of maximum matchings}

\label{SECbasics}

Recall that 
$G=(U,V;E)$ denotes a bipartite graph in which every edge in $e\in E$ connects a vertex in $U$ and one in $V$, 
and 
an edge set $M\subseteq E$ is a \emph{matching} if each vertex in $U\cup V$ is adjacent to 
at most one edge in $M$. 
Let $M\subseteq E$ be a matching in $G=(U,V;E)$. 
A vertex is \emph{matched by $M$} if it is incident to an edge in $M$, 
and \emph{unmatched by $M$} otherwise. 
For a vertex $z\in U\cup V$, 
let $M(z)$ denote the vertex matched to $z$ by $M$. 
Note that $M(z)$ is empty if $z$ is unmatched. 
In this case we let $M(z)=\emptyset$. 

A vertex subset $C \subseteq U \cup V$ is a \emph{vertex cover} if 
each edge in $E$ is incident to at least one vertex in $C$. 
The following theorem provides a fundamental min-max formula of matchings and vertex covers. 
\begin{theorem}[K\H{o}nig \cite{Kon31}]
\label{THMkonig}
    In a bipartite graph, 
    the maximum size of a matching is equal to the minimum size of a vertex cover. 
\end{theorem}

Let $M \subseteq E$ be a matching in $G$. 
An \emph{alternating path with respect to $M$} is a path in which the edges in $M$ and $E \setminus M$ appear alternately. 
An \emph{alternating cycle with respect to $M$} is defined in the same way. 
They may be simply called \emph{alternating path/cycle} if $M$ is clear from the context. 
The length of a path and a cycle is defined by the number of its edges. 

Let $z\in U \cup V$ be a vertex and $M\subseteq E$ be a maximum matching in $G$. 
If there exists an alternating path 
with respect to $M$ of even (resp.,\ odd) length connecting $z$ and a vertex unmatched by $M$, 
then $z$ is called \emph{even} (resp.,\ \emph{odd}). 
If such a path does not exist, 
then $z$ is called \emph{unreachable}. 
On the basis of the Dulmage-Mendelsohn decomposition \cite{DM58,DM59}, 
we have that 
these definitions are independent of the choice of the maximum matching $M$, 
and the following lemma holds. 
\begin{lemma}
\label{LEMDM}
    Let $z$ be a vertex and $M$ be a maximum matching in a bipartite graph. 
    \begin{enumerate}
        \item If $z$ is odd or unreachable, then $v$ is matched by $M$. 
        \label{ENUDM1}
        \item If $z$ is odd, then $M(z)$ is even. 
        \label{ENUDM2}
        \item If $z$ is unreachable, then $M(z)$ is unreachable. 
        \label{ENUDM3}
\end{enumerate}
\end{lemma}

If edge weights $w(e)\in \RR$ $(e\in E)$ are associated to $G$, 
then the weighted graph is denoted by $(G,w)$, 
where $w\in \RR^E$ is a vector consisting of the weights $w(e)$ ($e\in E$). 
The following linear program with variable $x\in \RR^E$ provides a linear relaxation of the 
maximum-weight matching problem in $(G,w)$: 
\begin{alignat}{3}
\label{EQLPmm_obj}
&\mbox{Maximize}    &\quad& \sum_{(u,v)\in E} w(u,v)\cdot x(u,v )     && \\
\label{EQLPmm_const1}
&\mbox{subject to}  &\quad& \sum_{v\in \Gamma(u)}x(u,v)\le 1                &\quad& (u\in U),\\
\label{EQLPmm_const2}
&                   &\quad& \sum_{u\in \Gamma(v)}x(u,v)\le 1             &\quad& (v\in V),\\
&                   &\quad&x(u,v)\ge 0                                  &\quad& ((u,v)\in E).
\label{EQLPmm_const3}
\end{alignat}

For an edge subset $M\subseteq E$, 
define the \emph{characteristic vector} $\chi_M \in \{0,1\}^E$ of $M$ by 
$\chi_M(e) = 1$ if $e\in M$ and $\chi_M(e) = 0$ if $e\in E\setminus M$.
It is straightforward to see that 
an integer vector $x\in \ZZ^E$ 
satisfies \eqref{EQLPmm_const1}--\eqref{EQLPmm_const3} 
if and only if $x=\chi_M$ for a matching $M$.

%%Note that the equality constraint \eqref{EQLPconst1} is replaced with an inequality constraint. 
The dual problem of 
\eqref{EQLPmm_obj}--\eqref{EQLPmm_const3} is the following linear program with variable $y\in \RR^{U\cup V} $.
\begin{alignat}{3}
\label{EQdualmm_obj}
&\mbox{Minimize}    &\quad& \sum_{u\in U} y(u) + \sum_{v\in V} y(v)     && \\
\label{EQdualmm_const1}
&\mbox{subject to}  &\quad& y(u) + y(v) \ge w(u,v)                    &\quad& ((u,v)\in E),\\
\label{EQdualmm_const2}
&                   &\quad& y(u)\ge 0                                   &\quad& (u\in U), \\
\label{EQdualmm_const3}
&                   &\quad& y(v)\ge 0                                   &\quad& (v\in V).
\end{alignat}

A vector $y\in \RR^{U\cup V}$ satisfying the constraints \eqref{EQdualmm_const1}--\eqref{EQdualmm_const3} is referred to as a 
\emph{$w$-vertex cover}, 
and 
a $w$-vertex cover is \emph{minimum} if it minimizes the objective function \eqref{EQdualmm_obj}. 

A fundamental theorem in combinatorial optimization states that the optimality of the linear program \eqref{EQLPmm_obj}--\eqref{EQLPmm_const3} is attained by an integral solution, 
and 
the dual problem \eqref{EQdualmm_obj}--\eqref{EQdualmm_const3} as well if $w$ is integral. 
Namely, 
the characteristic vector of a maximum-weight matching is an optimal solution of \eqref{EQLPmm_obj}--\eqref{EQLPmm_const3}, 
and 
there exists a minimum $w$-vertex cover which is integral. 
This fact can be derived from the total unimodularity of the incidence matrix of a bipartite graph, i.e.,\ 
the coefficient matrix of the linear system \eqref{EQLPmm_const1}--\eqref{EQLPmm_const2}, 
while it is specifically known as Egerv\'ary's theorem \cite{Ege31} for bipartite matchings.

\begin{lemma}[\cite{Ege31}, see also \protect{\cite[Theorem 17.1 and Section 18.3]{Sch03}}]
\label{LEMintmm}
The linear program \eqref{EQLPmm_obj}--\eqref{EQLPmm_const3} has an integral optimal solution. 
Moreover, 
if $w$ is integral, 
then the dual problem \eqref{EQdualmm_obj}--\eqref{EQdualmm_const3} also has an integral optimal solution. 
\end{lemma}

\subsection{The house allocation problem}

In describing the house allocation problem (HA problem), 
we often use the terminology of \emph{applicants} and  
\emph{houses}, 
to emphasize that the vertices on only one side have preferences over the other side. 
Let $G=(A,H;E)$ be a bipartite graph, 
where $A$ denotes the set of the applicants and 
$H$ that of the houses. 
Without loss of generality, 
assume that no vertex is isolated. 
Recall that
$\Gamma (u)$ denotes the set of vertices adjacent to $u$ 
for each vertex $u\in A \cup H$. 
Namely, 
$\Gamma (a) =\{h\in H \colon (a,h)\in E\}$ for each $a\in A$ and 
$\Gamma (h) =\{a\in A \colon (a,h)\in E\}$ for each $h\in H$.

Each applicant $a\in A$ has preferences over the adjacent houses, 
which is described by a totally ordered set $(\Gamma (a) \cup \{\emptyset\}, \succ_a)$. 
If $h\succ_a h'$ holds for two houses $h,h' \in \Gamma(a)$, 
it means that $a$ prefers $h$ to $h'$. 
% Recall that $\emptyset$ means that $a$ is unmatched, 
% and assume that $h\succ_a \emptyset$ for each $a\in A$ and each $h\in \Gamma(a)$. 
Recall that $h\succ_a \emptyset$ holds for each $a\in A$ and each $h\in \Gamma(a)$. 

In order 
to describe the case in which $a$ is unmatched  
in the same manner as the case in which $a$ is matched, 
for each applicant $a \in A$,  
we 
add a dummy house $l(a)$ to $H$, 
and an an edge $(a,l(a))$ to $E$. 
The dummy house $l(a)$ is referred to as the \emph{last resort} of $a$, 
and is assumed to be the house least preferred by $a$. 
Now the totally ordered set $(\Gamma (a) \cup \{\emptyset\}, \succ_a)$ is simply denoted by 
$(\Gamma (a), \succ_a)$, 
and hereafter we only discuss matchings in which every applicant $a\in A$ is matched. 
In order to emphasize that each applicant in $A$ is matched, 
a matching is typically referred to as \emph{$A$-perfect}. 

Now an instance of the HA problem is represented by a tuple 
\begin{align*}
    (G=(A,H;E),(\Gamma(a),\succ_a)_{a\in A}).
\end{align*}
For two matchings $M$ and $N$ in $G$, 
let $\Delta(M,N)$ denote the number of applicants who prefer $M$ to $N$, 
formally, 
\begin{align*}
\Delta (M,N) = |\{a\in A \colon M(a) \succ_a N(a)\}| - |\{a\in A \colon N(a) \succ_a M(a)\}|. 
\end{align*}
A matching $M$ is \emph{popular} if 
$\Delta(M,N) \ge 0$ for each matching $N$ in $G$.

% A primary algorithmic issue of popular matchings is to 
% provide a characterization of popular matchings which can be checked efficiently. 
% This is indeed substantial because simply applying the above definition requires 
% calculation of $\Delta(M,N)$ and $\Delta(N,M)$ for 
% all matchings $N$ in $G$, 
% which may be of an exponential number. 

The graph-structural characterization of popular matchings in the HA problem \cite{AIKM07} is described as follows. 
Let $a\in A$. 
Let $f(a) \in \Gamma(a)$ denote the house most preferred by $a$, 
and define $H_f\subseteq H$ by 
\begin{align*}
H_f=\bigcup_{a\in A}\{f(a)\}=\{h\in H\colon \mbox{$h=f(a)$ for some $a\in A$}\}.     
\end{align*}
Let $s(a)$ denote the house in $\Gamma(a) \setminus H_f$ most preferred by $a$. 
Note that $s(a)$ always exists due to the addition of the last resort $l(a)$. 

\begin{theorem}[\cite{AIKM07}]
    \label{THMgraphHA}
    In an instance $(G=(A,H;E),(\Gamma(a),\succ_a)_{a\in A})$ of the HA problem, 
    a matching $M\subseteq E$ is popular if and only if it satisfies the following two conditions. 
    \begin{enumerate}
        \item Each house $h \in H_f$ is matched by $M$. 
            \label{ENUgraphHA1}
        \item For each applicant $a\in A$, it holds that $M(a)\in \{f(a),s(a)\}$. 
            \label{ENUgraphHA2}
    \end{enumerate}
\end{theorem}

The optimization-based characterization \cite{KMN11,McC08} 
is described in the following way. 
For a matching $M$ in $G$, 
define 
an edge-weight vector $w_M\in \{0,1,2\}^E$ by 
\begin{align}
\label{EQweightHA}
    w_M(a,h)=
    \begin{cases}
        2 & (h \succ_a M(a)), \\
        1 & (h = M(a)), \\
        0 & (M(a) \succ_a h). 
    \end{cases}
\end{align}
% For an edge set $F \subseteq E$, 
% its weight $w(F)$ is defined by $w(F) = \sum_{e\in F}w(e)$. 

\begin{theorem}[\cite{KMN11,McC08}]
    \label{THMopt}
    In an instance $(G=(A,H;E),(\Gamma(a),\succ_a)_{a\in A})$ of the HA problem, 
    a matching $M\subseteq E$ is popular if and only if $M$ is a maximum-weight $A$-perfect matching in $(G,w_M)$ 
    with weight vector $w_M$ defined by \eqref{EQweightHA}. 
\end{theorem}

\subsection{The house allocation problem with ties}
    
The house allocation problem with ties (HAT problem) is defined by 
allowing for ties in the preferences in the HA problem. 
The preferences of an applicant $a\in A$ is represented by $(\Gamma(a), \succsim_a)$. 
Here, $h \succ_a h'$ means that $a$ strictly prefers a house $h\in \Gamma(a)$ to another house $h' \in\Gamma(a)$, 
and 
$h \sim_a h'$ that $a$ is indifferent between $h,h'\in \Gamma(a)$. 
Also, 
$h \succsim_a h'$ means $h \succ_a h'$ or $h \sim_a h'$. 
Now an instance of the HAT problem is represented by a tuple
$(G=(A,H;E),(\Gamma(a),\succsim_a)_{a\in A})$. 
The definition of a popular matching is straightforwardly extended.

The graph-structural characterization of a popular matching in the HA problem (Theorem \ref{THMgraphHA}) is extended by utilizing the Dulmage-Mendelsohn decomposition \cite{DM58,DM59}. 
For each applicant $a\in A$, 
let $f(a)\subseteq H$ denote the set of houses most preferred by $a$. 
Note that $f(a)$ denotes a set of houses, 
while it has denoted one house in the HA problem. 
Construct a bipartite graph $G_f = (A, H; E_f)$, 
where 
\begin{align*}
    E_f= \{ (a,h) \colon a\in A, h\in f(a) \}. 
\end{align*}
For a matching $M\subseteq E$, 
let $M_f = M \cap E_f$. 

We now apply the Dulmage-Mendelsohn decomposition to $G_f$ to partition the vertices in $A\cup H$ 
into those that are even, those that are odd, and those that are unreachable. 
% Let $M^* \subseteq E_f$ be an arbitrary maximum matching in $G_f$. 
% An \emph{alternating path} is a path in which the edges in $E_f \setminus M^*$ and $M$ appear alternately. 
% The length of a path is defined by the number of its edges. 
% A vertex $v$ in $G_f$ is called \emph{even} (resp.,\ \emph{odd}) if there exists an alternating path of even (resp.,\ odd) length to $v$ starting from a vertex unmatched by $M^*$. 
% A vertex $v$ in $G_f$ with no alternating path reaching $v$ from a vertex unmatched by $M^*$ is called \emph{unreachable}. 
% We remark that these definitions is independent of the choice of the maximum matching $M^*$. 
For an applicant $a\in A$, 
define $s(a)\subseteq H$ as 
the set of even houses most preferred by $a$. 
Note that 
$s(a) \neq \emptyset$ because the last resort $l(a)$ is always isolated in $G_f$ and hence even. 
Also observe that possibly $s(a)\subseteq f(a)$, 
and otherwise $f(a) \cap s(a) = \emptyset$. 

\begin{theorem}[\cite{AIKM07}]
    \label{THMgraphHAT}
    In an instance $(G=(A,H;E),(\Gamma(a),\succsim_a)_{a\in A})$ of the HAT problem, 
    a matching $M\subseteq E$ is popular if and only if it satisfies the following two conditions. 
    \begin{enumerate}
        \item $M_f$ is a maximum matching in $G_f$.
        \label{ENUgraphHAT1}
        \item For each applicant $a\in A$, it holds that $M(a)\in f(a)\cup s(a)$. 
        \label{ENUgraphHAT2}
    \end{enumerate}
\end{theorem}

An optimization-based characterization of a popular matching in the HAT problem is obtained from Theorem \ref{THMopt} by 
involving the ties in the preferences into the definition of the weight vector $w_M$. 
Let $M$ be a matching in $G$. 
Define a vector $w_M\in \{0,1,2\}^E$ by 
\begin{align}
    \label{EQweightHAT}
    w_M(a,h)=
    \begin{cases}
        2 & (h \succ_a M(a)), \\
        1 & (h \sim_a M(a)), \\
        0 & (M(a) \succ_a h). 
    \end{cases}
\end{align}

\begin{theorem}[\cite{KMN11,McC08}]
    \label{THMoptHAT}
    In an instance $(G=(A,H;E),(\Gamma(a),\succsim_a)_{a\in A})$ of the HAT problem, 
    a matching $M\subseteq E$ is popular if and only if $M$ is a maximum-weight $A$-perfect matching in $(G,w_M)$ 
    with weight vector $w_M$ defined by \eqref{EQweightHAT}. 
\end{theorem}

\subsection{The stable matching problem with incomplete lists}

The stable matching problem with incomplete lists (SMI problem) consists of two sides of agents, 
each having preferences. 
Let $G=(U,V;E)$ be a bipartite graph, 
in which each vertex in $U\cup V$ represents an agent. 
Again, 
assume that no vertex is isolated. 
Each agent $u\in U\cup V$ has strict preferences $(\Gamma(u) \cup \{\emptyset\}, \succ_u)$, 
in which $\emptyset$ is the least preferred. 
Note that we do not add last resorts in the SMI problem. 
An instance of the SMI problem is represented by a tuple 
\begin{align*}
    (G=(U,V;E),(\Gamma(z)\cup \{\emptyset\},\succ_z)_{z\in U\cup V}). 
\end{align*}

For two matchings $M,N\subseteq E$ in $G$, 
define $\Delta(M,N)\in \ZZ$ by 
\begin{align*}
    \Delta (M,N) = |\{z\in U\cup V \colon M(z) \succ_u N(z)\}| - |\{z\in U\cup V \colon N(z) \succ_z M(z)\}|. 
\end{align*}
A matching $M$ is \emph{popular} if 
$\Delta(M,N)\ge 0$ holds for each matching $N$ in $G$. 

Huang and Kavitha \cite{HK13} gave the following characterization of popular matchings in the SMI problem. 
Let 
\begin{align*}
    (G=(U,V;E),(\Gamma(z) \cup \{\emptyset\},\succ_z)_{z\in U\cup V})
\end{align*}
be an instance of the SMI problem, 
and let $M\subseteq E$ be a matching in $G$. 
For each edge $e=(u,v)\in E$, 
where $u\in U$ and $v\in V$, 
define its \emph{label} $(\alpha_M(e),\beta_M(e))$ by 
% \begin{align}
%     \label{EQalpha}
%     &{}\alpha_M(e) = 
%     \begin{cases}
%         + & (\mbox{$u$ is matched by $M$ and $v\succ_u M(u)$, or $u$ is unmatched by $M$}), \\
%         0 & (e\in M), \\
%         - & (\mbox{$u$ is matched by $M$ and $M(u)\succ_u v$)}, 
%     \end{cases} \\
%     \label{EQdelta}
%     &{}\beta_M(e) = 
%     \begin{cases}
%         + & (\mbox{$v$ is matched by $M$ and $u\succ_v M(v)$, or $v$ is unmatched by $M$}), \\
%         0 & (e\in M), \\
%         - & (\mbox{$v$ is matched by $M$ and $M(v)\succ_v u$)}. 
%     \end{cases}
% \end{align}
\begin{align}
    \label{EQalpha}
    &{}\alpha_M(e) = 
    \begin{cases}
        + & (\mbox{$u$ is matched and $v\succ_u M(u)$, or $u$ is unmatched}), \\
        0 & (e\in M), \\
        - & (\mbox{$u$ is matched and $M(u)\succ_u v$)}, 
    \end{cases} \\
    \label{EQdelta}
    &{}\beta_M(e) = 
    \begin{cases}
        + & (\mbox{$v$ is matched and $u\succ_v M(v)$, or $v$ is unmatched}), \\
        0 & (e\in M), \\
        - & (\mbox{$v$ is matched and $M(v)\succ_v u$)}. 
    \end{cases}
\end{align}
Let $G_M^+=(U, V; E_M^+)$ be a graph obtained from $G$ by deleting every edge $e\in E$ with 
$(\alpha_M(e), \beta_M(e))=(-,-)$. 

\begin{theorem}[\cite{HK13}]
    \label{THMgraphSMI}
    In an instance $(G=(U,V;E),(\Gamma(u)\cup \{\emptyset\},\succ_u)_{u\in U\cup V})$ of the SMI problem, 
    a matching $M\subseteq E$ in $G$ is popular if and only if it satisfies the following conditions. 
    \begin{enumerate}
        \item $G_M^+$ has no alternating cycle with respect to $M$ including an edge $e$ with label $(\alpha_M(e),\beta_M(e))=(+,+)$. 
        \label{ENUgraphSMI1}
        \item $G_M^+$ has no alternating path with respect to $M$ including an endpoint unmatched by $M$ and an edge $e$ with label $(\alpha_M(e),\beta_M(e))=(+,+)$. 
        \label{ENUgraphSMI2}
        \item $G_M^+$ has no alternating path with respect to $M$ including two or more edges $e$ with label $(\alpha_M(e),\beta_M(e))=(+,+)$. 
        \label{ENUgraphSMI3}
    \end{enumerate}
\end{theorem}

The optimization-based characterization for the SMI problem  is described in the following way. 
Let $M\subseteq E$ be a matching in $G$ 
and $e=(u,v)\in E$, 
where $u\in U$ and $v\in V$. 
Define the label $(\phi_M(e), \psi_M(e))$ of $e$ as  
\begin{align}
    &{}\phi_M(e)=
    \begin{cases}
        2 & (\mbox{$u$ is matched and $v \succ_u M(u)$}),\\
        1 & (\mbox{$e\in M$ or $u$ is unmatched}),\\
        0 & (\mbox{$u$ is matched and $M(u) \succ_u v$}), 
    \end{cases}
    \label{EQweightSMI1}\\
    &{}\psi_M(e) = 
    \begin{cases}
        2 & (\mbox{$v$ is matched and $u \succ_v M(v)$}),\\
        1 & (\mbox{$e\in M$ or $v$ is unmatched}),\\
        0 & (\mbox{$v$ is matched and $M(v) \succ_v u$}). 
    \end{cases}
    \label{EQweightSMI2}
\end{align}
Then, 
the weight $w_M(e)$ is defined by 
\begin{align}
    &{}w_M(e) = \phi_M(e) + \psi_M(e). 
    \label{EQweightSMI3}
\end{align}
Note that 
$(+,0,-)$ in the label $(\alpha_M(e), \beta_M(e))$
does not exactly correspond to
$(2,1,0)$ in $(\phi_M(e), \psi_M(e))$.

\begin{theorem}[\cite{BIM10}]
    \label{THMoptSMI}
    In an instance $(G=(U,V;E),(\Gamma(z)\cup \{\emptyset\},\succ_z)_{z\in U\cup V})$ of the SMI problem, 
    a matching $M\subseteq E$ in $G$ is popular if and only if it is a maximum-weight matching in $(G,w_M)$ 
    with weight vector $w_M$ defined by \eqref{EQweightSMI1}--\eqref{EQweightSMI3}. 
 \end{theorem}

\section{Equivalence of the characterizations for the HA problem}
\label{SECHA}

In this section, 
we prove that the two characterizations in Theorems \ref{THMgraphHA} and \ref{THMopt} are equivalent 
without using the fact that $M$ is popular. 

\subsection{Deriving the optimization-based characterization for the HA problem}

Let $(G=(A,H;E),(\Gamma(a),\succ_a)_{a\in A})$ be an instance of the HA problem. 
Suppose that a matching $M\subseteq E$ in $G$ satisfies 
the graph-structural characterization, 
i.e.,\ 
Properties \ref{ENUgraphHA1} and \ref{ENUgraphHA2} in 
Theorem \ref{THMgraphHA}. 
We prove that $M$ is a maximum-weight $A$-perfect matching in the weighted graph $(G,w_M)$, 
where the edge-weight vector $w_M\in \{0,1,2\}^E$ is defined by \eqref{EQweightHA}.

The following linear program with variable $x\in \RR^E$ is a linear relaxation of the 
maximum-weight $A$-perfect matching problem in the weighted graph $(G,w_M)$: 
\begin{alignat}{3}
\label{EQLPobj}
&\mbox{Maximize}    &\quad& \sum_{(e)\in E} w_M(e)\cdot x(e)     && \\
&\mbox{subject to}  &\quad& \sum_{h\in \Gamma(a)}x(a,h)=1                &\quad& (a\in A),
\label{EQLPconst1} \\
&                   &\quad& \sum_{a\in \Gamma(h)}x(a,h)\le 1             &\quad& (h\in H),
\label{EQLPconst2}\\
&                   &\quad&x(e)\ge 0                                        &\quad& (e\in E).
\label{EQLPconst3}
\end{alignat}
The following linear program with variable $y\in \RR^{A\cup H} $ is the dual problem of \eqref{EQLPobj}--\eqref{EQLPconst3}: 
\begin{alignat}{3}
\label{EQdualobj}
&\mbox{Minimize}    &\quad& \sum_{a\in A} y(a) + \sum_{h\in H} y(h)     && \\
\label{EQdualconst1}
&\mbox{subject to}  &\quad& y(a) + y(h) \ge w_M(e)                    &\quad& (e=(a,h)\in E),\\
\label{EQdualconst2}
&                   &\quad& y(h)\ge 0                                   &\quad& (h\in H).
\end{alignat}

The following lemma is straightforward. 
\begin{lemma}
\label{LEMprimeHA}
    If $M$ is an $A$-perfect matching, 
    then 
    the characteristic vector $\chi_M$ of $M$ is a feasible solution of the linear program \eqref{EQLPobj}--\eqref{EQLPconst3} with objective value $|A|$. 
\end{lemma}

Below we prove that $\chi_M$ is an optimal solution for \eqref{EQLPobj}--\eqref{EQLPconst3} by 
constructing a feasible solution of the dual problem \eqref{EQdualobj}--\eqref{EQdualconst2} with objective value $|A|$. 

It follows from Property \ref{ENUgraphHA2} of Theorem \ref{THMgraphHA} that 
the set $A$ of applicants is partitioned into two sets $A_f$ and $A_s$, 
where 
\begin{align*}
    A_f= \{ a\in A\colon M(a) = f(a) \}, \quad 
    A_s= \{ a\in A\colon M(a) = s(a) \}. 
\end{align*}
On the basis of this partition, 
define $y^* \in \{0,1\}^{A\cup H}$ by
\begin{align}
    \label{EQyHA}
    y^*(a) = 
    \begin{cases}
        0 & (a\in A_f), \\
        1 & (a\in A_s),
    \end{cases}
    \quad 
    y^*(h) = 
    \begin{cases}
        1 & (h\in H_f), \\
        0 & (h\in H\setminus H_f).
    \end{cases}
\end{align}

\begin{lemma}
\label{LEMdualHA}
    If a matching $M\subseteq E$ in $G$ satisfies Properties \ref{ENUgraphHA1} and \ref{ENUgraphHA2} in Theorem \ref{THMgraphHA}, 
    then the vector $y^*$ defined by \eqref{EQyHA} is a feasible solution of the dual problem \eqref{EQdualobj}--\eqref{EQdualconst2} with objective value $|A|$. 
\end{lemma}

\begin{proof}
    We first prove the feasibility of $y^*$ defined by \eqref{EQyHA}. 
    The constraint \eqref{EQdualconst2} is clearly satisfied. 
    Below we show that \eqref{EQdualconst1} is satisfied for each edge $e=(a,h)\in E$, 
    where $a\in A$ and $h\in H$. 

    Suppose that $a\in A_f$, 
    i.e.,\ 
    $M(a)=f(a)$. 
    Note that $y^*(a)=0$. 
    If $h=M(a)$, 
    we have that 
    $w_M(e)=1$ and $h\in H_f$, 
    implying that $y^*(h)=1$. 
    If $h\in \Gamma(a)  \setminus \{ M(a)\}$, 
    it follows from the definition of $f(a)$ that $w_M(e)=0$. 
    Hence \eqref{EQdualconst1} is satisfied in both cases. 

    Next suppose that $a\in A_s$, 
    i.e.,\ 
    $M(a)=s(a)$. 
    It follows that $y^*(a)=1$. 
    We have the following three cases, 
    in each of which \eqref{EQdualconst1} is satisfied. 
        \paragraph{Case 1: $h=s(a)$} We have that $w_M(e)=1$. 
        It follows from the definition of $s(a)$ that $h\in H \setminus H_f$, and hence from \eqref{EQyHA} that $y^*(h)=0$. 
        Thus, 
        \eqref{EQdualconst1} holds with equality. 
        \paragraph{Case 2: $h\in H_f$} It follows from \eqref{EQyHA} that $y^*(h)=1$. 
        We thus derive \eqref{EQdualconst1} from $w_M(e)\le 2$. 
        \paragraph{Case 3: $h\in H\setminus (H_f \cup \{s(a)\})$}
        It follows from the definition of $H_f$ and $s(a)$ that $s(a)\succ_a h$, 
        and hence $w_M(e)=0$. 
        Thus, \eqref{EQdualconst1} is satisfied regardless of the value of $y^*(h)$, 
        while we know that $y^*(h)=0$. 

    We have shown that $y^*$ is a feasible solution of \eqref{EQdualobj}--\eqref{EQdualconst2}. 
    It follows from \eqref{EQyHA} that its objective value of $y^*$ is equal to 
    \begin{align}
    \sum_{a\in A} y^*(a) + \sum_{h\in H} y^*(h) = |A_s| + |H_f|. 
    \label{EQobjHA}
    \end{align}
    We complete the proof by showing that $|H_f| = |A_f|$. 
    It follows from the definitions of $A_f$ and $H_f$ that each applicant in $A_f$ is matched with a house in $H_f$ by $M$, 
    implying that $|H_f|\ge |A_f|$. 
    Suppose to the  contrary that there exists a house $h\in H_f$ with which no applicant in $A_f$ is matched. 
    It then follows from Property \ref{ENUgraphHA1} in Theorem \ref{THMgraphHA} that $M(h)\in A_s$. 
    This contradicts the fact that $h\in H_f$ and the definition of $A_s$. 
    We thus obtain $|H_f| = |A_f|$, 
    completing the proof. 
\end{proof}

On the basis of the strong duality theorem for linear optimization, 
we obtain from Lemmas \ref{LEMprimeHA} and \ref{LEMdualHA} that 
$\chi_M$ is an optimal solution of \eqref{EQLPobj}--\eqref{EQLPconst3}, 
implying that $M$ is a maximum-weight $A$-perfect matching in the weighted graph $(G,w_M)$.

The above proof implies a new connection of the structure of popular matchings to the theory of combinatorial optimization. 
We conclude this subsection with the following corollaries. 

\begin{corollary}
    \label{CORwcoverHA}
    Let $(G=(A,H;E),(\Gamma(a),\succ_a)_{a\in A})$ be an instance of the HA problem. 
    In the weighted graph $(G,w_M)$ 
    in which $w_M$ is defined by \eqref{EQweightHA}, 
    the vector $y^*$ defined by \eqref{EQyHA} is a minimum $w_M$-vertex cover. 
\end{corollary}

Namely, 
if a matching $M$ satisfies Properties \ref{ENUgraphHA1} and \ref{ENUgraphHA2} in Theorem \ref{THMgraphHA} 
(i.e.,\ if $M$ is a popular matching), 
then 
a minimum $w_M$-vertex cover is obtained from $A_s$ and $H_f$, 
which are defined from the preferences of the applicants 
without reference to either the linear program \eqref{EQLPobj}--\eqref{EQLPconst3} or \eqref{EQdualobj}--\eqref{EQdualconst2}.
It is also noteworthy that a minimum $w_M$-vertex cover is obtained as a $\{0,1\}$-vector, 
%%This is trivial if $w_M$ is also a $\{0,1\}$-vector, 
which does not directly follow from Lemma \ref{LEMintmm} since 
%%but in our problem 
$w_M$ is a $\{0,1,2\}$-vector. 

\begin{corollary}
    \label{CORHA}
    The dual problem \eqref{EQdualobj}--\eqref{EQdualconst2} in which $w_M$ is defined by \eqref{EQweightHA} has a $\{0,1\}$-optimal solution. 
\end{corollary}

\subsection{Deriving the graph-structural characterization for the HA problem}
\label{SECHA2}

In this subsection, 
we present the converse derivation. 
Let 
\begin{align*}
(G=(A,H;E),(\Gamma(a),\succ_a)_{a\in A})    
\end{align*} 
be an instance of the HA problem, 
and 
suppose that a matching $M$ in $G$ is a maximum-weight $A$-perfect matching in the weighted graph $(G,w_M)$, 
where $w_M$ is defined by \eqref{EQweightHA}. 
We prove that $M$ satisfies Properties \ref{ENUgraphHA1} and \ref{ENUgraphHA2} in Theorem \ref{THMgraphHA}. 

We begin with stating the basic facts in Section \ref{SECbasics} in terms of maximum-weight $A$-perfect matchings. 
In the same way as for Lemma \ref{LEMintmm}, 
it is not difficult to derive the integrality of the linear programs \eqref{EQLPobj}--\eqref{EQLPconst3} 
and \eqref{EQdualobj}--\eqref{EQdualconst2} from the total unimodularity of 
the coefficient matrix 
and the fact that $w_M$ is an integral vector. 
\begin{lemma}
\label{LEMint}
The linear programs \eqref{EQLPobj}--\eqref{EQLPconst3} and \eqref{EQdualobj}--\eqref{EQdualconst2} have an integral optimal solution. 
\end{lemma}

It follows from Lemma \ref{LEMint} and the fact that $M$ is a maximum-weight $A$-perfect matching that 
the characteristic vector $\chi_M$ of $M$ is an optimal solution of the linear program \eqref{EQLPobj}--\eqref{EQLPconst3}. 
Again on the basis of Lemma \ref{LEMint}, 
let $y^*\in \ZZ^{A\cup H}$ denote an integral optimal solution of 
the dual problem \eqref{EQdualobj}--\eqref{EQdualconst2}. 

Note that $\chi_M(e) > 0$ implies $e\in M$, and hence $w_M(e)=1$. 
Now 
the complementary slackness condition of these linear programs is that 
\begin{alignat}{2}
    \label{EQCS1}
    &{}\mbox{$\chi_M(e) > 0$ implies $y^*(a)+y^*(h)= w_M(e)=1$}& \quad &(e=(a,h)\in E),\\
    \label{EQCS2}
    &{}\mbox{$y^*(h) > 0$ implies $\displaystyle\sum_{a\in \Gamma(h)}\chi_M(a,h)= 1$}& \quad& (h\in H).
\end{alignat}

We begin an argument specific to our problem with the following lemma, 
stating that an integral optimal solution $y^*$ of \eqref{EQdualobj}--\eqref{EQdualconst2} is a $\{0,1\}$-vector. 
Note that this fact is not trivial because 
$w_M(e)$ may be equal to two for $e\in E$ and 
$y^*(a)$ may be negative for $a\in A$. 

\begin{lemma}
\label{LEMy01}
    For an integral optimal solution $y^* \in \ZZ^{A\cup H}$ of \eqref{EQdualobj}--\eqref{EQdualconst2}, 
    it holds that $y^* \in \{0,1\}^{A\cup H}$. 
\end{lemma}

\begin{proof}
    We first show that $y^*(a)\in \{0,1\}$ for each applicant $a\in A$. 
    Let $a\in A$ and $h=M(a)$. 
    It follows from $(a,h)\in M$ and \eqref{EQCS1} that $y^*(a)+y^*(h)=1$. 
    From \eqref{EQdualconst2}, we obtain that 
    \begin{align}
        \label{EQyale1}
        y^*(a)\leq 1. 
    \end{align} 

    Suppose that $h\neq l(a)$. 
    We have that $(a,l(a)) \not\in M$ and $l(a)$ is only adjacent to $a$ in $G$, 
    and hence 
    \begin{align*}
        \sum_{a'\in \Gamma(l(a))}\chi_M(a',l(a))= 0. 
    \end{align*}
    It then follows from \eqref{EQCS2} that 
    \begin{align}
    \label{EQyla}
        y^*(l(a))=0. 
    \end{align}
    We also have that 
    %%$h\in \Gamma(a) \setminus \{ l(a)\}$, 
    %%and hence 
    $h \succ_a l(a)$, 
    implying that $w_M(a,l(a))=0$. 
    It then follows from \eqref{EQdualconst1} that 
    \begin{align}
        \label{EQ3.11}
        y^*(a) + y^*(l(a)) \ge 0.
    \end{align}
    From \eqref{EQyla} and \eqref{EQ3.11}, 
    we obtain $y^*(a)\ge 0$. 
    We thus derive from \eqref{EQyale1} 
    that 
    $y^*(a)\in \{0,1\}$, 
    since $y^*$ is integral. 
    From \eqref{EQCS1}, 
    we also obtain $y^*(M(a)) \in \{0,1\}$. 

    Suppose that $h=l(a)$. 
    Let $h'\in \Gamma(a) \setminus \{l(a)\}$, 
    which must exist since $a$ is not isolated before $l(a)$ is attached. 
    We have that $h' \succ_a l(a)=M(a)$, and hence $w_M(a,h')=2$. 
    It follows from \eqref{EQdualconst1} that 
    \begin{align}
        \label{EQ3.14}
        y^*(a) + y^*(h')\ge 2, 
    \end{align}
    and further 
    \begin{align}
        \label{EQyh}
        y^*(h')\ge 1
    \end{align} from \eqref{EQyale1}. 
    It then follows from \eqref{EQCS2} that $\sum_{a\in \Gamma(h')}\chi_M(a,h')= 1$, 
    implying that $(a',h')\in M$ for some applicant $a'\in A$. 
    We can now apply the above argument to $a'$ to obtain $y^*(h')=y^*(M(a'))\in \{0,1\}$, 
    and hence $y^*(h')= 1$ from \eqref{EQyh}. 
    We thus obtain $y^*(a)= 1$ from \eqref{EQyale1} and \eqref{EQ3.14}. 

    We have shown that $y^*(a)\in \{0,1\}$ for each $a\in A$. 
    We complete the proof by showing that $y^*(h)\in \{0,1\}$ for each $h\in H$. 
    Let $h\in H$. 
    If $(a,h)\in M$ for some $a\in A$, 
    it follows from \eqref{EQCS1} that $y^*(h)=1-y^*(a)\in \{0,1\}$. 
    Otherwise, 
    $\sum_{a\in \Gamma(h)}\chi_M(a,h)= 0$, 
    and 
    we obtain $y^*(h)=0$ from \eqref{EQCS2}. 
    %%We thus conclude that $y^*(h)\in \{0,1\}$ for each $h\in H$. 
\end{proof}

On the basis of Lemma \ref{LEMy01}, 
we partition the set $A$ of applicants into two sets $A_0$ and $A_1$ defined by 
\begin{align*}
    A_0 = \{a\in A \colon y^*(a) =0\}, \quad A_1 = \{a\in A \colon y^*(a) =1\}.
\end{align*}

\begin{lemma}
    \label{LEMA0HA}
    For each $a\in A_0$, 
    it holds that $M(a)=f(a)$. 
\end{lemma}

\begin{proof}
    Let $a\in A_0$ and $h \in \Gamma(a) \setminus \{M(a)\}$. 
    It follows from \eqref{EQdualconst1} and Lemma \ref{LEMy01} that 
    $w_M(a,h)\le y^*(a) + y^*(h) \le 0+1=1$. 
    It thus holds that $M(a)\succ_a h$, 
    which implies that $M(a)=f(a)$. 
\end{proof}

\begin{lemma}
    \label{LEMA1HA}
    For each applicant $a\in A_1$, 
    it holds that $M(a) \in \{f(a),s(a)\}$. 
\end{lemma}

\begin{proof}
    Let $a\in A_1$ and 
    suppose that $M(a)\neq f(a)$. 
    We prove the lemma by showing that $M(a) = s(a)$. 

    Let $h\in \Gamma(a)$ satisfy $h\succ_a M(a)$. 
    It follows from the definition \eqref{EQweightHA} of $w_M$ that $w_M(a,h) = 2$. 
    It then follows from \eqref{EQdualconst1} that 
    $y^*(h) \ge w_M(a,h) - y^*(a) =1$, 
    and hence from Lemma \ref{LEMy01} that $y^*(h)  =1$. 
    From \eqref{EQCS2}, 
    we obtain that $(a',h)\in M$ for some $a'\in A$. 
    From \eqref{EQCS1}, 
    we derive $y^*(a') = 1 - y^*(h) = 0$ and hence $a'\in A_0$. 
    By applying Lemma \ref{LEMA0HA}, 
    we have that $h=f(a')$ and hence $h\in H_f$. 

    We next show that $M(a) \not\in H_f$. 
    Suppose to the contrary that $M(a) = f(a'')$ for some applicant $a''\in A$. 
    We have that $(a'',M(a))\not\in M$, 
    and thus $w_M(a'',M(a))=2$. 
    It follows from \eqref{EQCS1} that $y^*(M(a))=1-y^*(a)=0$, 
    and from \eqref{EQdualconst1} that $y^*(a'')\ge w_M(a'',M(a)) - y^*(M(a))=2$, contradicting Lemma \ref{LEMy01}. 
    %%We thus have $M(a)\not\in H_f$. 

    Therefore, 
    we have shown that $h\in H_f$ for each $h\in \Gamma(a)$ satisfying $h\succ_a M(a)$, and 
    $M(a)\not \in H_f$. 
    We hence conclude that $M(a)=s(a)$. 
\end{proof}

Lemmas \ref{LEMA0HA} and \ref{LEMA1HA} amount to Property \ref{ENUgraphHA2} in Theorem \ref{THMgraphHA}. 
We finally derive Property \ref{ENUgraphHA1} in Theorem \ref{THMgraphHA}. 

\begin{lemma}
    Each house $h\in H_f$ is matched by $M$. 
\end{lemma}

\begin{proof}
    Let $h\in H$ be a house not matched by $M$. 
    It suffices to show that $h\not\in H_f$. 
    We have seen in the proof of Lemma \ref{LEMy01} that $y^*(h)=0$ if $h$ is not matched by $M$. 
    Let $a\in \Gamma(h)$. 
    It follows from \eqref{EQCS1} and Lemma \ref{LEMy01} that $w_M(a,h)\le y^*(a)+y^*(h) \le 1$. 
    Since $(a,h)\not \in M$, 
    this implies that $w_M(a,h) = 0$. 
    We thus obtain that $h\neq f(a)$ for each $a\in \Gamma(h)$, and hence $h\not\in H_f$. 
\end{proof}

We have now proved that 
a maximum-weight $A$-perfect matching $M$ in the weighted graph $(G,w_M)$ 
satisfies Properties \ref{ENUgraphHA1} and \ref{ENUgraphHA2} in Theorem \ref{THMgraphHA}.

\section{Equivalence of the characterizations for the HAT problem}
\label{SECHAT}

In this section, 
we prove that the two characterizations in Theorems \ref{THMgraphHAT} and \ref{THMoptHAT} are equivalent 
without using the fact that $M$ is popular.

\subsection{Deriving the optimization-based characterization  for the HAT problem}

Let $(G=(A,H;E),(\Gamma(a),\succsim_a)_{a\in A})$ be an instance of the 
HAT problem. 
Suppose that a matching $M\subseteq E$ satisfies 
the graph-structural characterization, 
i.e.,\ 
Properties \ref{ENUgraphHAT1} and \ref{ENUgraphHAT2} in 
Theorem \ref{THMgraphHAT}. 
We prove that $M$ is a maximum-weight $A$-perfect matching in the weighted graph $(G,w_M)$, 
where the edge-weight vector $w_M\in \{0,1,2\}^E$ is defined by \eqref{EQweightHAT}. 

We prove that the characteristic vector $\chi_M \in \{0,1\}^E$ of $M$ 
is an optimal solution for \eqref{EQLPobj}--\eqref{EQLPconst3} by 
constructing a feasible solution of the dual problem \eqref{EQdualobj}--\eqref{EQdualconst2} with objective value $|A|$. 
By using Property \ref{ENUgraphHAT2} of Theorem \ref{THMgraphHAT}, 
we partition the set $A$ of the applicants into two sets $A_f$ and $A_s$, 
where 
\begin{align*}
    A_f = \{a\in A \colon M(a) \in f(a)\}, \quad A_s = \{a\in A \colon M(a) \in s(a) \setminus f(a)\}. 
\end{align*}
Recall that the vertices in $A\cup H$ 
are partitioned into odd, even, and unreachable vertices by the Dulmage-Mendelsohn  decomposition applied to $G_f$.
Define $y^* \in \{0,1\}^{A\cup H}$ by 
\begin{alignat}{2}
    &{}y^*(a) = 
    \begin{cases}
        0 & \mbox{($a$ is unreachable, or $a$ is even and $a\in A_f$)}, \\
        1 & \mbox{($a$ is odd, or $a$ is even and $a\in A_s$)} 
    \end{cases}
    &\quad& (a\in A),
    \label{EQyaHAT}
    \\
    &{}y^*(h) = 
    \begin{cases}
        0 & \mbox{($h$ is even)}, \\
        1 & \mbox{($h$ is odd or unreachable)} 
    \end{cases}
    &\quad& (h\in H).    
    \label{EQyhHAT}
\end{alignat}

\begin{lemma}
    \label{LEMdualHAT}
    If a matching $M\subseteq E$ in $G$ satisfies Properties \ref{ENUgraphHAT1} and \ref{ENUgraphHAT2} in Theorem \ref{THMgraphHAT}, 
    then the vector $y^*\in \{0,1\}^{A\cup H}$ defined by \eqref{EQyaHAT} and \eqref{EQyhHAT} is a feasible solution of the dual problem \eqref{EQdualobj}--\eqref{EQdualconst2} 
    with objective value $|A|$. 
\end{lemma}

\begin{proof}
    We first show that $y^*$ defined by \eqref{EQyaHAT} and \eqref{EQyhHAT} satisfies the constraints 
    \eqref{EQdualconst1} and \eqref{EQdualconst2}. 
    The constraint \eqref{EQdualconst2} is clearly satisfied. 
    Below we show that \eqref{EQdualconst1} is satisfied for each edge $(a,h)\in E$, 
    where $a\in A$ and $h\in H$.

    % \paragraph{Case 1: $a\in A_f$.}
    % We have that $M(a)\in f(a)$. 
    % Thus, if $(a,h)\in E \setminus E_f$, 
    % it follows that $w_M(a,h)=0$ and hence we are done. 
    % Below we assume $(a,h)\in E_f$. 
    % It then follows that $w_M(a,h)=1$. 

    \paragraph{Case 1: $a\in A_f$}
    We have that $M(a)\in f(a)$. 
    Thus, if $(a,h)\in E \setminus E_f$, 
    it follows that $w_M(a,h)=0$ and hence we are done. 
    Below we assume $(a,h)\in E_f$. 
    It then follows that $w_M(a,h)=1$. 
    \paragraph{Case 1.1: $a\in A_f$ is even}
    In this case, 
    it follows that $h$ is odd, 
    and thus
    $y^*(h)=1$. 

    \paragraph{Case 1.2: $a\in A_f$ is odd}
    In this case, 
    we have that $y^*(a)=1$. 

    \paragraph{Case 1.3: $a\in A_f$ is unreachable}
    In this case, 
    $h$ is odd or unreachable, and thus $y^*(h)=1$. 

    Therefore, 
    in Cases 1.1--1.3, 
    we have that $y^*(a)+y^*(h)\ge 1 = w_M(a,h)$. 

    \paragraph{Case 2: $a\in A_s$}
    We have that $M(a)\in s(a)\setminus f(a)$, implying that $(a,M(a))\not \in E_f$. 
    In particular $a$ is unmatched by $M_f$, 
    and hence $a$ is even. 
    It thus follows from \eqref{EQyaHAT} that $y^*(a)=1$.

    Without loss of generality 
    we assume that $h \succ_a M(a)$, 
    because otherwise $w_M(a,h)\le 1 \le y^*(a)+y^*(h)$ holds 
    regardless whether $y^*(h)=0$ or $y^*(h)=1$. 
    It then follows from $M(a)\in s(a)$ and the definition of $s(a)$ that $h$ is odd or unreachable. 
    We thus conclude that $y^*(h)=1$, 
    $y^*(a)=1$, and $w_M(a,h)=2$, 
    and hence we are done.

    We now show that the objective value of $y^*$ is equal to $|A|$. 
    It follows from Properties \ref{ENUDM1} and \ref{ENUDM3} in Lemma \ref{LEMDM} that 
    the number of unreachable applicants in $A$ is equal to that of unreachable houses in $H$. 
    It also follows from Properties \ref{ENUDM1} and \ref{ENUDM2} in Lemma \ref{LEMDM} that 
    the number of odd houses in $H$ is equal to that of even applicants in $A_f$. 
    Hence we have that $|\{a\in A\colon y^*(a)=0\}| = |\{h\in H\colon y^*(h)=1\}| $. 
    We thus conclude that the objective value of $y^*$ is 
    \begin{align*}
        \sum_{a\in A} y^*(a) + \sum_{h\in H} y^*(h)  
        {}&{}= |\{a\in A\colon y^*(a)=1\}| + |\{h\in H\colon y^*(h)=1\}| \\
        {}&{}= |\{a\in A\colon y^*(a)=1\}| + |\{a\in A\colon y^*(a)=0\}| 
        = |A|,
    \end{align*}
    completing the proof. 
\end{proof}

On the basis of the strong duality theorem for linear optimization, 
we obtain from Lemmas \ref{LEMprimeHA} and \ref{LEMdualHAT} that 
$\chi_M$ is an optimal solution of \eqref{EQLPobj}--\eqref{EQLPconst3}, 
implying that $M$ is a maximum-weight $A$-perfect matching in the weighted graph $(G,w_M)$. 

Corresponding to Corollary \ref{CORwcoverHA}, 
a corollary on a minimum $w_M$-cover is obtained as follows. 

\begin{corollary}
\label{CORwcoverHAT}
    Let $(G=(A,H;E),(\Gamma(a),\succsim_a)_{a\in A})$ be an instance of the HAT problem. 
    In the weighted graph $(G,w_M)$ 
    in which $w_M$ is defined by \eqref{EQweightHAT}, 
    the vector $y^*$ defined by \eqref{EQyaHAT} and \eqref{EQyhHAT} is a minimum $w_M$-vertex cover. 
\end{corollary}

An interpretation of Corollary \ref{CORwcoverHAT} is as follows. 
If a matching $M$ in $G$ satisfies Properties \ref{ENUgraphHAT1} and \ref{ENUgraphHAT2} in Theorem \ref{THMgraphHAT} 
(i.e.,\ if $M$ is a popular matching), 
then 
a minimum $w_M$-vertex cover is obtained from the partition $\{A_f,A_s\}$ of $A$ and the Dulmage-Mendelsohn decomposition 
of $G_f$. 
Similarly to those for the HA problem, 
these are defined 
without reference to either the linear program \eqref{EQLPobj}--\eqref{EQLPconst3} or \eqref{EQdualobj}--\eqref{EQdualconst2}.
Further, 
again a minimum $w_M$-vertex cover is obtained as a $\{0,1\}$-vector, 
while $w_M$ is a $\{0,1,2\}$-vector. 

\begin{corollary}
    \label{CORHAT}
    The dual problem \eqref{EQdualobj}--\eqref{EQdualconst2} in which $w_M$ is defined by \eqref{EQweightHAT} has a $\{0,1\}$-optimal solution. 
\end{corollary}

\subsection{Deriving the graph-structural characterization for the HAT problem}

Let $(G=(A,H;E),(\Gamma(a),\succsim_a)_{a\in A})$ be an instance of the 
HAT problem, 
and suppose that a matching $M$ in $G$ is a maximum-weight $A$-perfect matching in the weighted graph $(G,w_M)$, 
where $w_M$ is defined by \eqref{EQweightHAT}. 
We prove that $M$ satisfies Properties \ref{ENUgraphHAT1} and \ref{ENUgraphHAT2} in Theorem \ref{THMgraphHAT}. 

In the same manner as in Section \ref{SECHA2}, 
we have that 
the dual problem \eqref{EQdualobj}--\eqref{EQdualconst2} admits an 
integral optimal solution $y^*\in \ZZ^{A\cup H}$
satisfying \eqref{EQCS1} and \eqref{EQCS2}. 
Further, 
it is straightforward to verify that 
Lemma \ref{LEMy01} applies to the HAT problem as well, 
i.e.,\ $y^*\in \{0,1\}^{A\cup H}$. 
On the basis of this fact, 
we partition each of $A$ and $H$ into two sets: 
\begin{align*}
    &{}A_0 = \{a\in A \colon y^*(a) =0\}, \quad A_1 = \{a\in A \colon y^*(a) =1\}, \\
    &{}H_0 = \{h\in H \colon y^*(h) =0\}, \quad H_1 = \{h\in H \colon y^*(h) =1\}. 
\end{align*}

\begin{lemma}
    %% Lemma 13
    \label{LEMA0H1}
    It holds that 
    $M(a) \in H_1$ for each $a\in A_0$ and 
    $M(h) \in A_0$ for each $h\in H_1$, 
    implying that $|A_0|=|H_1|$. 
\end{lemma}

\begin{proof}
    Let $a\in A_0$. 
    Since $M$ is an $A$-perfect matching, 
    we have that $M(a)$ exists. 
    It follows from the complementarity slackness \eqref{EQCS1} that 
    $ y^*(M(a))=1 - y^*(a) = 1$, 
    and hence $M(a)\in H_1$. 

    Let $h\in H_1$. 
    It follows from 
    the complementarity slackness \eqref{EQCS2} that 
    $h$ is matched by $M$. 
    Now the same argument derives that $M(h)\in A_0$. 
\end{proof}

The next lemma derives Property \ref{ENUgraphHAT2} in Theorem \ref{THMgraphHAT} 
for the applicants in $A_0$. 

\begin{lemma}
    % Lemma 14
    \label{LEMA0}
    For each $a\in A_0$, 
    it holds that $M(a)\in f(a)$. 
\end{lemma}

\begin{proof}
Let $a\in A_0$ and let $h$ be an arbitrary house in $\Gamma(a)$. 
It follows from \eqref{EQdualconst1} that 
\begin{align*}
    w_M(a,h) \le y^*(a)+y^*(h)=y^*(h)\le 1.
\end{align*} 
We thus obtain from \eqref{EQweightHAT} that 
$M(a)\sim_a h$ or $M(a)\succ_a h$, 
implying that $M(a)\in f(a)$. 
\end{proof}

We then derive Property \ref{ENUgraphHAT1} in Theorem \ref{THMgraphHAT}. 
Define  $A'_1\subseteq A_1$ 
by 
%%as a set of applicants in $A_1$
%%who are connected to a house in $H_1$ by an edge in $E_f$ but not to a house in $H_0$, 
%%i.e.,\ 
\begin{align}
    %%A_1' = \{ a\in A_1\colon \mbox{$(a,h)\in E_f$ for some $h\in H_1$ and $(a,h)\not\in E_f$ for each $h\in H_0$} \}. 
    %%A_1' = \{a\in A_1 \colon f(a)\subseteq H_1\}. 
    A_1' = \{a\in A_1 \colon f(a) \cap H_0 \neq \emptyset\}.
    \label{EQA1prime}
\end{align}
Below we show that $A_1' \cup H_1$ is a vertex cover of $G_f$ satisfying that 
$|A_1' \cup H_1|=|M_f|$. 

\begin{lemma}
\label{LEMcover}
    It holds that $A_1' \cup H_1$ is a vertex cover in $G_f=(A,H;E_f)$. 
\end{lemma}

\begin{proof}
    Let $(a,h)\in E_f$, i.e.,\ $h\in f(a)$. 
    It suffices to show that $a\in  A\setminus A_1'$ implies $h\in H_1$ and 
    $h\in H_0$ implies $a\in A_1'$. 

    First 
    suppose that $a\in A_1\setminus A_1'$. 
    It follows from the definition \eqref{EQA1prime} of $A_1'$ that 
    $f(a)\cap H_0=\emptyset$. 
    Then $h\in H_1$ follows from $h\in f(a)$. 

    Next 
    suppose that $a \in A_0$. 
    It follows from Lemma \ref{LEMA0} that $(a,h) \sim_a M(a)$ and hence  $w_M(a,h)=1$. 
    It then follows from \eqref{EQdualconst1} that $y^*(h) \ge w_M(a,h)-y^*(a)=1$ and hence $h\in H_1$. 

    Finally, 
    suppose that $h\in H_0$. 
    The above argument implies that $a\not\in A_0$, namely $a\in A_1$. 
    Then, 
    $a\in A_1'$ follows from 
    \eqref{EQA1prime} and 
    the fact that 
    $h\in f(a) \cap H_0$. 
\end{proof}

\begin{lemma}
\label{LEMmf}
    $|M_f| = |A_1'| +|H_1|$. 
\end{lemma}

\begin{proof}
    Let $a\in A$. 
    Since $M$ is an $A$-perfect matching, 
    we have that $(a,M(a))$ exists. 
    Below we investigate whether $(a,M(a))\in M_f$. 

    If $a\in A_0$, it follows from Lemma \ref{LEMA0} that $M(a)\in f(a)$ and hence $(a,M(a))\in M_f$. 

    Suppose that $a\in A_1'$. 
    It follows from the definition \eqref{EQA1prime} of $A_1'$ that 
    there exists a house $h\in f(a) \cap H_0$. 
    It follows from $a\in A_1'$, $h\in H_0$ and \eqref{EQdualconst1} that $w_M(a,h) \le y^*(a)+y^*(h)=1+0=1$. 
    It also follows from $h\in f(a)$ that $w_M(a,h)\ge 1$. 
    We thus obtain $w_M(a,h)=1$ and $M(a)\in f(a)$, 
    i.e.,\ $(a,M(a))\in M_f$. 
    
    Suppose that $a\in A_1\setminus A_1'$. 
    It follows from \eqref{EQCS1} that $y^*(M(a))=1-y^*(a)=0$, 
    and hence $M(a)\in H_0$. 
    It then follows from the definition \eqref{EQA1prime} of $A_1'$ that 
    $f(a)\cap H_0=\emptyset$, 
    and thus $M(a)\not\in f(a)$, 
    i.e.,\ 
    $(a,M(a))\not \in M_f$. 

    Therefore, 
    we conclude that $|M_f|=|A_0|+|A_1'|=|A_1'| +|H_1|$, 
    where the latter equality follows from Lemma \ref{LEMA0H1}. 
\end{proof}

On the basis of Theorem \ref{THMkonig}, 
we obtain 
from Lemmas \ref{LEMcover} and \ref{LEMmf}
that $M_f$ is a maximum matching in $G_f$. 
Finally, 
we derive Property \ref{ENUgraphHAT2} in Theorem \ref{THMgraphHAT} 
for the applicants in $A_1$. 

\begin{lemma}
    % Lemma 18
    \label{LEMA1}
    For each $a\in A_1$, 
    it holds that $M(a)\in f(a) \cup s(a)$. 
\end{lemma}

\begin{proof}
    In the proof of Lemma \ref{LEMmf}, 
    we have shown that $M(a)\in f(a)$ if $a\in A_1'$, 
    and 
    $M(a)\not\in f(a)$ if $a\in A_1 \setminus A_1'$. 
    We prove the lemma by showing that $M(a)\in s(a)$ for each $a\in A_1 \setminus A_1'$. 

    Let $a\in A_1 \setminus A_1'$. 
    We have that $(a,M(a)) \not \in M_f$, 
    and hence $M(a)$ is even. 
    % Since $M(a)\not\in f(a)$, 
    % there exists a house $h\in \delta a$ such that $h\succ_a M(a)$. 
    Below we show that each house $h\in \Gamma(a)$ such that $h\succ_a M(a)$ is not even, 
    which completes the proof. 

    Let $h\in \Gamma(a)$ and $h\succ_a M(a)$. 
    It follows that $w_M(a,h)=2$, 
    implying that $y^*(h)\ge w_M(a,h) - y^*(a)=1$, and hence $h\in H_1$. 
    
    Suppose to the contrary that $h$ is even. 
    Since $M_f$ is a maximum matching in $G_f$, 
    there exists an alternating path $P$ of even length connecting $h$ and a house $h'$ unmatched by $M_f$. 
    It follows from Lemmas \ref{LEMA0H1} and \ref{LEMA0} that each house in $H_1$ is matched by $M_f$, 
    and hence $h'\in H_0$. 
    Let $P=(h_0,a_0,h_1,a_1,\ldots, h_{k},a_k,h_{k+1})$, 
    where 
    $h_0=h'$, $h_{k+1}=h$, 
    $(h_i,a_i)\in E_f \setminus M_f$ 
    and 
    $(a_i,h_{i+1})\in M_f$ for each $i=0,1,\ldots, k$. 
    We have that $w_M(e)=1$ for each edge $e$ in $P$, 
    and hence $y^*(a_0) \ge w_M(h_0,a_0) - y^*(h_0)=1-0=1$, 
    implying that $y^*(a_0)=1$. 
    It then follows from \eqref{EQCS1} that 
    $y^*(h_1)=w_M(a_0,h_1)-y^*(a_0)=1-1=0$. 
    By repeating this argument, 
    we obtain $y^*(h)=0$, 
    contradicting that $h\in H_1$.    
    We thus conclude that $h$ is not even. 
\end{proof}

\section{Equivalence of the characterizations for the SMI problem}
\label{SECSMI}
In this section, 
we prove that the two characterizations in Theorems \ref{THMgraphSMI} and \ref{THMoptSMI} are equivalent 
without using the fact that $M$ is popular.

\subsection{Deriving the optimization-based characterization for the SMI problem}

Let $(G=(U,V;E),(\Gamma(z)\cup \{\emptyset\},\succ_z)_{z\in U\cup V})$ be an instance of the SMI problem. 
Suppose that a matching $M\subseteq E$ in $G$ satisfies the graph-structural characterization, 
i.e.,\ Properties \ref{ENUgraphSMI1}--\ref{ENUgraphSMI3} in Theorem \ref{THMgraphSMI}. 
We prove that $M$ is a maximum-weight matching in $(G,w_M)$, 
where the weight vector $w_M$ is defined by \eqref{EQweightSMI1}--\eqref{EQweightSMI3}.

In this section, 
we refer to a linear program obtained from \eqref{EQLPmm_obj}--\eqref{EQLPmm_const3} by  replacing $w$ with $w_M$
simply as \eqref{EQLPmm_obj}--\eqref{EQLPmm_const3}. 
The same applies to the dual problem \eqref{EQdualmm_obj}--\eqref{EQdualmm_const3}. 

It follows from the definition of $w_M$ that 
$w_M(e)=2$ for each edge $e\in M$ 
and hence the characteristic vector $\chi_M$ of $M$ has objective value $2|M|$ in 
the linear program \eqref{EQLPmm_obj}--\eqref{EQLPmm_const3}. 
Below we construct a feasible solution $y^*\in \RR^{U\cup V}$ of the dual problem \eqref{EQdualmm_obj}--\eqref{EQdualmm_const3} 
with objective value $2|M|$. 

Let $\mathcal{P}$ be a family of all alternating paths with respect to $M$ in $G_M^+$ including an edge $e$ 
with label $(\alpha_M(e),\beta_M(e))=(+,+)$, 
and let $P\in \mathcal{P}$. 
It follows from Property \ref{ENUgraphSMI3} of Theorem \ref{THMgraphSMI} that $P$ includes a unique edge $e$ with $(\alpha_M(e),\beta_M(e))=(+,+)$, 
which is denoted by $e=(u_P,v_P)$, 
where $u_P\in U$ and $v_P\in V$. 
It then follows from Property \ref{ENUgraphSMI2} of Theorem \ref{THMgraphSMI} that both $u_P$ and $v_P$ are matched by $M$, 
and hence $\phi_M(e)=\psi_M(e)=2$, 
i.e.,\ $w_M(u_P,v_P)=4$.

Every vertex on $P$ can be classified according to its distance to $u_P$ and $v_P$ along $P$. 
Specifically, 
a vertex $z \in U\cup V$ in $P$ is said to be \emph{closer to $u_P$ than $v_P$} if a subpath of $P$ connecting $z$ and $u_P$ does not 
include $v_P$, 
and otherwise $z$ is \emph{closer to $v_P$ than $u_P$}. 
%%This concept naturally partitions the vertices on $P$ into two sets: those closer to $u_P$, and those closer to $v_P$.

On the basis of this concept, 
define $\Ueven, \Uodd \subseteq U$ and $\Veven, \Vodd \subseteq V$ by 
\begin{align*}
    &{}\Ueven = \{u\in U \colon \mbox{$u$ is closer to $u_P$ for some $ P\in \mathcal{P}$}\}, \\
    &{}\Uodd = \{u\in U \colon \mbox{$u$ is closer to $v_P$ for some $ P\in \mathcal{P}$}\}, \\
    &{}\Veven = \{v\in V \colon \mbox{$v$ is closer to $v_P$ for some $ P\in \mathcal{P}$}\}, \\
    &{}\Vodd = \{v\in V \colon \mbox{$v$ is closer to $u_P$ for some $ P\in \mathcal{P}$}\}.
\end{align*}
It holds that these four sets give a partition of the set of vertices included in the alternating paths in $\mathcal{P}$. 

\begin{lemma}
    \label{LEMdisjoint}
    It holds that $\Uodd \cap \Ueven = \emptyset$ and $\Vodd \cap \Veven = \emptyset$.
\end{lemma}

\begin{proof}
    Below we prove that $\Uodd \cap \Ueven = \emptyset$. 
    The same argument proves $\Vodd \cap \Veven = \emptyset$. 

    Suppose to the contrary that $u\in \Uodd \cap \Ueven$ for some vertex $u\in U$. 
    It follows that $u$ is closer to $u_P$ for some $P\in \mathcal{P}$ and 
    closer to $v_{P'}$ for some $P'\in \mathcal{P}$. 
    Let $P(u,v_P)$ denote a subpath of $P$ connecting $u$ and $v_P$, 
    and 
    $P'(u,u_{P'})$
    that of $P'$ connecting $u$ and $u_{P'}$. 
    It follows from $(u_P,v_P),(u_{P'},v_{P'})\not\in M$ that 
    the edge in $P(u,v_P)$ incident to $u$ does not belong to $M$, 
    but that in $P'(u,u_{P'})$ does. 
    Thus, 
    concatenating $P(u,v_P)$ and $P'(u,u_{P'})$ obtains a new alternating  
    $\tilde{P}$ connecting $v_P$ and $u_{P'}$. 
    It follows that $\tilde{P}$ is an alternating path including two edges $(u_P, v_P)$ and $(u_{P'}, v_{P'})$ 
    with label $(+,+)$, 
    contradicting Property \ref{ENUgraphSMI3} in Theorem \ref{THMgraphSMI}. 
    We thus conclude that $\Uodd \cap \Ueven = \emptyset$.     
\end{proof}

On the basis of Lemma \ref{LEMdisjoint}, 
define $y^*\in \{0,1,2\}^{U\cup V}$ as follows. 
\begin{itemize}
    \item If $z \in U\cup V$ is included in  some alternating path $P\in \mathcal{P}$, then 
    \begin{align}
    \label{EQySMI1}
    y^*(z) = 
    \begin{cases}
        2 & (z\in \Ueven\cup \Veven), \\
        0 & (z\in \Uodd \cup \Vodd). 
    \end{cases}
    \end{align}

    \item 
    If $z \in U\cup V$ is not included by any alternating path $P\in \mathcal{P}$, then 
    \begin{align}
    \label{EQySMI2}
    y^*(z) = 
    \begin{cases}
        1 & (\mbox{$z$ is matched by $M$}), \\
        0 & (\mbox{$z$ is unmatched by $M$}). 
    \end{cases}
    \end{align}
\end{itemize}

\begin{lemma}
    The vector $y^*\in \{0,1,2\}^{U\cup V}$ 
    defined by \eqref{EQySMI1} and \eqref{EQySMI2} is 
    a feasible solution of the dual problem \eqref{EQdualmm_obj}--\eqref{EQdualmm_const3} 
    with objective value $2|M|$. 
\end{lemma}

\begin{proof}
    Let $y^*\in \{0,1,2\}^{U\cup V}$ be 
    defined by \eqref{EQySMI1} and \eqref{EQySMI2}. 
    We first show the feasibility of $y^*$. 
    It is clear that $y^*$ satisfies \eqref{EQdualmm_const2} and \eqref{EQdualmm_const3}. 
    Here we show that $y^*$ satisfies \eqref{EQdualmm_const1}. 
    Let $e=(u,v)$, 
    where $u\in U$ and $v\in V$. 
    It is clear if $w_M(u,v)=0$. 

    \paragraph{Case 1: $w_M(e)=1$}
    Without loss of generality, assume that $\phi_M(e)=1$ and $\psi_M(e)=0$. 
    It follows from the definitions \eqref{EQweightSMI1} and \eqref{EQweightSMI2} of $\phi_M$ and $\psi_M$ that 
    $u$ is unmatched while $v$ is matched by $M$. 

    Suppose to the contrary that \eqref{EQdualmm_const1} is violated, 
    namely, 
    $y^*(u)=y^*(v)=0$. 
    We have that $v$ is matched, 
    and hence $v\in \Vodd$ follows from \eqref{EQySMI1} and \eqref{EQySMI2}. 
    This implies that $u\in \Ueven$ and hence $y^*(u)=2$, 
    contradicting $y^*(u)=0$. 

    \paragraph{Case 2: $w_M(e) = 2$}
    First suppose that $(\phi_M(e), \psi_M(e))=(2,0),(0,2)$. 
    Without loss of generality, 
    assume that $(\phi_M(e), \psi_M(e))=(2,0)$. 
    It follows from \eqref{EQalpha}--\eqref{EQweightSMI2} that 
    both $u$ and $v$ are matched by $M$ while $e\not\in M$, 
    $(\alpha_M(e),\beta_M(e))=(+,-)$, 
    and hence $e\in E_M^+$.

    Assume to the contrary that $y^*(u)+y^*(v)<2$. 
    If $(y^*(u),y^*(v))=(1,0)$, 
    it follows that $v\in \Vodd$ since $v$ is matched by $M$. 
    It then follows from $(u,v)\in E_M^+$ that $u\in \Ueven$ and hence $y^*(u)=2$, 
    contradicting that $y^*(u)=1$. 
    The case $(y^*(u),y^*(v))=(0,1)$ can be discussed in the same way. 
Finally, assume that $(y^*(u),y^*(v))=(0,0)$.
It follows that $u\in U_{\mathrm{odd}}$ and $v\in V_{\mathrm{odd}}$.
As in the case above, $v\in V_{\mathrm{odd}}$ implies that
$u\in U_{\mathrm{even}}$, contradicting Lemma~5.1.

    Next suppose that $(\phi_M(e), \psi_M(e))=(1,1)$. 
    We have that either $(u,v)\in M$ or both $u$ and $v$ are unmatched. 
    
    Consider the former case. 
    We have that 
    $(\alpha_M(e),\beta_M(e))=(0,0)$ and hence $e\in E_M^+$. 
    If $u$ is included in some alternating path $P\in \mathcal{P}$, 
    then so is $v$. 
    It then follows that $(y^*(u),y^*(v))$ is $(2,0)$ or $(0,2)$. 
    Otherwise, 
    neither $u$ nor $v$ is included in any alternating path in $\mathcal{P}$, 
    and hence $(y^*(u),y^*(v))=(1,1)$. 
    
    Consider the latter case. 
    It follows that $(\alpha_M(e),\beta_M(e))=(+,+)$. 
    We thus derive that $u\in \Ueven$ and $v\in \Veven$, 
    and hence $(y^*(u),y^*(v))=(2,2)$. 
    Therefore, 
    in all cases \eqref{EQdualmm_const1} is satisfied.

    \paragraph{Case 3: $w_M(e) = 3$}
    Without loss of generality, 
    assume that $(\phi_M(e),\psi_M(e))=(2,1)$. 
    It follows that $v\succ_u M(u)$, $v$ is unmatched, and hence $(\alpha_M(e),\beta_M(e))=(+,+)$. 
    Then we again obtain that $(y^*(u),y^*(v))=(2,2)$ and \eqref{EQdualmm_const1} is satisfied. 
    
    \paragraph{Case 4: $w_M(u,v) = 4$}
    From an argument similar to that for Case 3, 
    we obtain that $(\alpha_M(e),\beta_M(e))=(+,+)$ and $(y^*(u),y^*(v))=(2,2)$. 

    \bigskip

    We have shown that $y^*$ is a feasible solution of the problem \eqref{EQdualmm_obj}--\eqref{EQdualmm_const3}. 
    Below we show that $\sum_{u\in U} y^*(u) + \sum_{v\in V} y^*(v) =2|M|$. 

    Let $U(\mathcal{P})$ denote the set of vertices in $U$ included in some path in $\mathcal{P}$, 
    and let $V(\mathcal{P})$ denote that for $V$. 
    Similarly, 
    let $M(\mathcal{P})$ denote the set of edges in $M$ included in a path in $\mathcal{P}$. 
    It follows from Property \ref{ENUgraphSMI2} in Theorem \ref{THMgraphSMI} that each vertex in $U(\mathcal{P}) \cup V(\mathcal{P})$ is matched by $M(\mathcal{P})$. 
    In particular, 
    each edge in $M(\mathcal{P})$ connects either vertices in $\Ueven$ and $\Vodd$, 
    or vertices in $\Veven$ and $\Uodd$. 
    It follows that $|\Ueven \cup \Veven|=|M(\mathcal{P})|$.

    % Let $u \in (U\cup V)\setminus (U(\mathcal{P})\cup  V(\mathcal{P}))$. 
    % If $u$ is matched by $M$, 
    % then $M(u) \in (U\cup V)\setminus (U(\mathcal{P})\cup  V(\mathcal{P}))$. 

    Observe that the set of vertices in$(U\cup V)\setminus (U(\mathcal{P})\cup  V(\mathcal{P}))$ matched by $M$ is 
    exactly the set of the endpoints of the edges in $M\setminus M(\mathcal{P})$, 
    and hence the number of those vertices is $2|M\setminus M(\mathcal{P})|$. 

    From the definition \eqref{EQySMI1} and \eqref{EQySMI2} of $y^*$, 
    we obtain that 
    \begin{align*}
        \sum_{u\in U} y^*(u) + \sum_{v\in V} y^*(v) 
        % &{}= 2 |\Ueven \cup \Veven| + |(U\cup V)\setminus (U(\mathcal{P})\cup  V(\mathcal{P}))\cap\partial M|\\
        % &{}= 2 |M(\mathcal{P})| + 2 |M\setminus M(\mathcal{P})| \\
        &{}= 2 |\Ueven \cup \Veven| + 2 |M\setminus M(\mathcal{P})|
        = 2|M|, 
    \end{align*}
    completing the proof. 
\end{proof}

Corresponding to Corollaries \ref{CORwcoverHA} and \ref{CORwcoverHAT}, 
a theorem on a minimum $w_M$-cover for the SMI problem is obtained as follows. 

\begin{corollary}
\label{CORwcoverSMI}
    Let $(G=(U,V;E),(\Gamma(z)\cup \{\emptyset\},\succ_z)_{z\in U\cup V})$ be an instance of the SMI problem. 
    In the weighted graph $(G,w_M)$ 
    in which $w_M$ is defined by \eqref{EQweightSMI1}--\eqref{EQweightSMI3}, 
    the vector $y^*$ defined by \eqref{EQySMI1} and \eqref{EQySMI2} is a minimum $w_M$-vertex cover. 
\end{corollary}

An interpretation of Corollary \ref{CORwcoverSMI} is as follows. 
If a matching $M$ in $G$ satisfies Properties \ref{ENUgraphSMI1}--\ref{ENUgraphSMI3} in Theorem \ref{THMgraphSMI} 
(i.e.,\ if $M$ is a popular matching), 
then 
a minimum $w_M$-vertex cover $y^*\in \{0,1,2\}^{U\cup V}$ is obtained from the alternating paths with respect to $M$ 
in $G_M^+$. 
Similarly to those for the HA and HAT problems, 
these are defined 
without reference to either the linear program \eqref{EQLPmm_obj}--\eqref{EQLPmm_const3} or \eqref{EQdualmm_obj}--\eqref{EQdualmm_const3}.
Further, 
a minimum $w_M$-vertex cover $y^*$ is obtained as a $\{0,1,2\}$-vector, 
while $w_M$ is a $\{0,1,2,3,4\}$-vector. 

\begin{corollary}
    \label{CORSMI}
    The dual problem \eqref{EQdualmm_obj}--\eqref{EQdualmm_const3} in which $w$ is replaced by $w_M$ defined by \eqref{EQweightSMI1}--\eqref{EQweightSMI3} has a $\{0,1,2\}$-optimal solution. 
\end{corollary}

\subsection{Deriving the graph-structural characterization for the SMI problem}

Let $(G=(U,V;E),(\Gamma(z)\cup \{\emptyset\},\succ_z)_{z\in U\cup V})$ be an instance of the SMI problem, 
and let $M$ be a matching in $G$ which violates at least one of Properties \ref{ENUgraphSMI1}--\ref{ENUgraphSMI3} 
in Theorem \ref{THMgraphSMI}. 
We prove that $M$ is not a maximum-weight matching in $(G, w_M)$ with $w_M$ defined by \eqref{EQweightSMI1}--\eqref{EQweightSMI3}. 

For each edge $e\in M$, 
it holds that $(\phi_M(e),\psi_M(e))=(1,1)$, and hence $w_M(e)=2$. 
It follows that $w_M(M)=2|M|$. 
Below we show that there exists a matching $M'$ with $w_M(M') > 2|M|$. 
For two edge sets $F_1,F_2\subseteq E$, 
let $F_1\triangle F_2$ denote their symmetric difference, 
i.e.,\ 
$F_1\triangle F_2=(F_1\setminus F_2) \cup (F_2\setminus F_1)$. 

\paragraph{Case 1: $M$ violates Property \ref{ENUgraphSMI1}}
Let $C$ be an alternating cycle with respect to $M$ in $G_M^+$ 
including an edge $e^*$ with label $(\alpha_M(e^*),\beta_M(e^*))=(+,+)$. 
By abuse of notation, 
let $C$ denote the set of edges in $C$. 
Note that $e^*\in C\setminus M$. 

It follows that every vertex in $C$ is matched by $M$, 
and hence each edge $e \in C\setminus M$ has label $(\phi_M(e),\psi_M(e))\in \{(2,2),(2,0),(0,2)\}$. 
Note that each edge $e$ with label $(\phi_M(e),\psi_M(e))=(0,0)$ has label $(\alpha_M(e),\beta_M(e)=(-,-)$ and hence is excluded in $G_M^+$. 
It follows that $w_M(e)\in \{2,4\}$ for each edge $e \in C\setminus M$, 
and in particular $w_M(e^*)=4$. 
We then obtain a matching $M' = M \triangle C$ 
with $w_M(M') \ge w_M(M)+2$.

\paragraph{Case 2: $M$ violates Property \ref{ENUgraphSMI2}}
Let $P$ be an alternating path with respect to $M$ in $G_M^+$ 
including an endpoint $z_1\in U\cup V$ unmatched by $M$ and an edge $e^*$ with $(\alpha_M(e^*),\beta_M(e^*))=(+,+)$. 
Again by abuse of notation, 
let $P$ denote the set of edges in $P$. 
Note that 
$e^*\in P\setminus M$. 

We assume that the other endpoint $z_2\in U\cup V$ of $P$ is not matched by an edge in $M\setminus P$, 
because otherwise we can extend $P$ by adding that edge. 
Namely, 
the endpoint $z_2$ is unmatched by $M$ 
or matched by an edge in $M\cap P$. 
It then follows that $M'=M\triangle P$ is a matching. 
Below we show that $w_M(M')> w_M(M)$. 
Denote the edge in $P$ including the endpoint $z_1$ by $e_1$, 
and that for $z_2$ by $e_2$. 

Suppose that $z_2$ is unmatched by $M$. 
It follows that $|P|=2k+1$ for some nonnegative integer $k$, 
where $|P\cap M| = k$ and $|P\setminus M|=k+1$. 
Note that $e_1,e_2\in P\setminus M$. 
For each edge $e\in P$, 
we can derive from \eqref{EQalpha}--\eqref{EQweightSMI3} that 
\begin{align*}
    w_M(e) 
    \begin{cases}
        =   2 & \mbox{if $e\in P\cap M$},\\
        \ge 1 & \mbox{if $e \in \{e_1,e_2\}$}, \\
        \ge 2 & \mbox{if $e\in P \setminus (M\cup \{e_1,e_2\}) $}. 
    \end{cases}
\end{align*}
Further, we have $e^*\in P \setminus M $ with 
$(\alpha_M(e^*),\beta_M(e^*))=(+,+)$. 
If $e^* \in \{e_1,e_2\}$, 
then $(\phi_M(e^*),\psi_M(e^*))\in \{(1,2),(2,1)\}$ and $w_M(e^*)= 3$. 
Otherwise, 
$(\phi_M(e^*),\psi_M(e^*)) = (2,2)$ and $w_M(e^*)=4$. 
In both cases, 
we have that $w_M(P\setminus M) \ge 2k+2$. 
It thus follows that 
\begin{align*}
    w_M(M') 
    {}&{}= w_M(M) - w_M(P\cap M) + w_M(P\setminus M) \\
    {}&{}\ge w_M(M) - 2k + (2k+2) = w_M(M) + 2.
\end{align*}

Next suppose that $z_2$ is matched by $M$. 
It follows that $|P|=2k$, 
where $|P\cap M| = k$ and $|P\setminus M|=k$, 
for some positive integer $k$. 
Note that $e_1\in P \setminus M$. 
For each edge $e\in P$, 
we can again derive from \eqref{EQalpha}--\eqref{EQweightSMI3} that 
\begin{align*}
    w_M(e) 
    \begin{cases}
        =   2 & \mbox{if $e\in P\cap M$},\\
        \ge 1 & \mbox{if $e = e_1$}, \\
        \ge 2 & \mbox{if $e\in P \setminus (M\cup \{e_1\}) $}. 
    \end{cases}
\end{align*}
Further, we have $e^*\in P \setminus M $ with 
$(\alpha_M(e^*),\beta_M(e^*))=(+,+)$. 
We can derive from the same argument as above that $w_M(P\setminus M)\ge 2k+1$. 
It thus follows that 
\begin{align*}
    w_M(M') 
    {}&{}= w_M(M) - w_M(P\cap M) + w_M(P\setminus M) \\
    {}&{}\ge w_M(M) - 2k + (2k+1) = w_M(M) + 1.
\end{align*}

\paragraph{Case 3: $M$ violates Property \ref{ENUgraphSMI3}}
Let $P$ be an alternating path with respect to $M$ in $G_M^+$ 
including two edges $e_1,e_2$ with label $(\alpha_M(\cdot),\beta_M(\cdot))=(+,+)$. 
If at least one endpoint of $P$ is unmatched by $M$, 
then it reduces to Case 2, 
and hence we can assume that both endpoints of $P$ are matched by $M$. 
It follows that 
$|P|=2k+1$, 
where $|P\cap M| = k+1$ and $|P\setminus M|=k$, 
for some positive integer $k$. 
For each edge $e\in P$, 
we can again derive from \eqref{EQalpha}--\eqref{EQweightSMI3} that 
\begin{align*}
    w_M(e) 
    \begin{cases}
        =   2 & \mbox{if $e\in P\cap M$},\\
        = 4 & \mbox{if $e\in \{e_1,e_2\}$}, \\
        \ge 2 & \mbox{if $e\in P \setminus (M \cup \{e_1,e_2\})$}. 
    \end{cases}
\end{align*}
% For each edge $e\in P\setminus M$, 
% it holds that $w_M(e)\ge 2$, 
% and $w_M(e_1)=w_M(e_2)=4$. 
It thus follows that 
\begin{align*}
    w_M(M') 
    {}&{}= w_M(M) - w_M(M\cap P) + w_M(M\setminus P) \\
    {}&{}\ge w_M(M) - 2(k+1) + (2 \cdot (k-2) + 4\cdot 2) = w_M(M) + 2.
\end{align*}

\bigskip 

Therefore, 
we have proved that 
a maximum-weight matching in $(G, w_M)$ with $w_M$ defined by \eqref{EQweightSMI1}--\eqref{EQweightSMI3} satisfies Properties \ref{ENUgraphSMI1}--\ref{ENUgraphSMI3} 
in Theorem \ref{THMgraphSMI}.

\section{Conclusion}
\label{SECcncl}

We have presented proofs of the equivalence of the graph-structural and optimization-based characterizations 
for the HA, HAT, and SMI problems of popular matchings. 
Our analysis suggests a potential approach for establishing graph-structural characterizations
for the problems for which only an optimization-based characterization is known. 
Representative targets include 
a graph-structural characterization for 
the SRTI problem, 
which might be obtained from the optimization-based characterization by utilizing the theory of nonbipartite matching. 

\bibliographystyle{abbrv}
\bibliography{refs}

\end{document}